\documentclass[aps,prd,groupedaddress,12pt,tightenlines,nofootinbib]{revtex4}
\usepackage{amsmath,amssymb,amsthm}

\usepackage{color}
\usepackage{fancybox}

\newtheorem{Theorem}{Theorem}
\newtheorem{Cor}{Corollary}
\newtheorem{Lemma}{Lemma}

\def\refa#1{(\ref{#1})}

\def\ra{\rightarrow}

\def\R{\mathbb{R}}
\def\RTX{\R_+^{1+3}}
\def\Linf{{L^\infty}}

\def\Linftx#1{{L^\infty_{1,#1}}}

\def\LinfTX{{L^\infty(\RTX)}}

\def\<{\langle}
\def\>{\rangle}
\def\norm#1{\<#1\>}

\def\d{\partial}
\def\lap{\Delta}

\def\la{\lambda}

\begin{document}

\title{Simple proof of a useful pointwise estimate for the wave equation}

\author{Nikodem Szpak}
\affiliation{Max-Planck-Institut f{\"u}r Gravitationsphysik, Albert-Einstein-Institut,
Golm, Germany}
\date{\today}

\begin{abstract}
We give a simple proof of a pointwise decay estimate stated in two versions, making advantage of a particular simplicity of inverting the spherically symmetric wave operator and of the comparison theorem. We briefly explain the role of this estimate in proving decay estimates for nonlinear wave equations or wave equations with potential terms.
\end{abstract}


\maketitle

\section{Introduction}

\subsection{Main result}

We consider a wave equation with  source
\begin{equation} \label{wave-eq}
   \Box \phi(t,x) = F(t,x)
\end{equation}
in 3 spatial dimensions, where $\Box\equiv \d_t^2 - \lap$, subject to null initial data $(\phi,\d_t \phi)(0,x)=(0,0)$.
We will use the notation $\norm{x}:=1+|x|$.
The main result of this note is a simple proof of the following very useful estimate
\begin{Lemma} \label{Lem:decay}
If 
\begin{equation} \label{source-bound}
   |F(t,x)|\leq \frac{A}{\norm{t+|x|}^p\norm{t-|x|}^q}
\end{equation}
with $p>2$, $q>1$ then
\begin{equation} 
   |\phi(t,x)|\leq \frac{C}{\norm{t+|x|}\norm{t-|x|}^{p-2}}
\end{equation}
with some positive constant $C$.
\end{Lemma}
It extends to a slightly more complicated and also very useful estimate
\begin{Lemma} \label{Lem:decay-x}
If 
\begin{equation} \label{source-bound-x}
   |F(t,x)|\leq \frac{A}{\norm{x}^\la\norm{t+|x|}^p\norm{t-|x|}^q}
\end{equation}
with $p>0$, $q>1, \la>2$ then
\begin{equation} 
   |\phi(t,x)|\leq \frac{C}{\norm{t+|x|}\norm{t-|x|}^{\nu}}
\end{equation}
with $\nu:=p+\min(q,\la-1)-1$ and some positive constant $C$.
\end{Lemma}
Observe that in case $q\geq \la-1$ we have $\nu=p+\la-2$ and lemma \ref{Lem:decay-x} implies lemma \ref{Lem:decay}, because condition \refa{source-bound-x} implies condition \refa{source-bound} with $p\ra p+\la$.

The techniques of both proofs are based on the well known comparison theorem.
\begin{Theorem} \label{Th:comparison}
If $\phi_1, \phi_2$ solve 
\begin{equation} 
   \Box \phi_1 = F_1,\qquad \Box \phi_2 = F_2
\end{equation}
with null initial data and $F_1(t,x)\leq F_2(t,x)$ for all $(t,x)\in\RTX$ then
\begin{equation} 
   \phi_1(t,x) \leq \phi_2(t,x)\qquad \forall (t,x)\in\RTX.
\end{equation}
\end{Theorem}
\begin{proof}
It follows from the fact that $\Box^{-1}$ is a measure on $\RTX$ and hence as an operator has a positive kernel. Then, $\phi_2-\phi_1 = \Box^{-1}(F_2-F_1) \geq 0$ when $F_2-F_1\geq 0$.
\end{proof}
Its immediate consequence is a bound
\begin{Cor} \label{Cor:comparison}
If $\phi_1, \phi_2$ solve 
\begin{equation} 
   \Box \phi_1 = F_1,\qquad \Box \phi_2 = F_2 \geq 0
\end{equation}
with null initial data and $|F_1(t,x)|\leq F_2(t,x)$ for all $(t,x)\in\RTX$ then
\begin{equation} 
   |\phi_1(t,x)| \leq \phi_2(t,x)\qquad \forall (t,x)\in\RTX.
\end{equation}
\end{Cor}
\begin{proof}
It is enough to use theorem \ref{Th:comparison} twice, once with $F_1(t,x)\leq F_2(t,x)$ and once with $-F_1(t,x)\leq F_2(t,x)$. We obtain then $\phi_1(t,x)\leq \phi_2(t,x)$ and $-\phi_1(t,x)\leq \phi_2(t,x)$ and hence $|\phi_1(t,x)|\leq \phi_2(t,x)$.

Observe, that one can use also theorem \ref{Th:comparison} with $0\leq F_2(t,x)$ to show that $0\leq \phi_2(t,x)$.
\end{proof}

Alinhac has used a similar technique in \cite{Alinhac-SemilinSys} for proving a different estimate of this type. Previously, in \cite{John-blowup, Asakura, Strauss-T, Georg-H-K, NS-WaveDecay}, more involved techniques have been used, based essentially on the integral (Duhamel) representation of the exact solution and estimation of the integrals over light-cones in space-time.

\subsection{Idea of the proof} \label{Sec:Idea}

The main idea of the proofs of lemma \ref{Lem:decay} and lemma \ref{Lem:decay-x} is that we estimate the source $|F(t,x)|$ by a source $G(t,r)$ which is spherically symmetric. For this source we can much more easily estimate the solutions $\varphi(t,r)$ of the corresponding equation \refa{wave-eq}, which are spherically symmetric as well. Finally, the obtained bound on $\varphi(t,r)$ presents, by corollary \ref{Cor:comparison}, a bound on $\phi(t,x)$, too. Below, we explain below the main steps of the proof of lemma \ref{Lem:decay}. The spirit of the proof of lemma \ref{Lem:decay-x} is the same.

In spherical symmetry (in 3 spatial dimensions) the wave equation reads
\begin{equation} \label{wave-eq-r}
   \d_t^2 \varphi - \frac{1}{r}\d_r^2(r \varphi) = G,
\end{equation}
with spherically symmetric source
\begin{equation} \label{source-bound-r}
   G(t,r) = \frac{A}{\norm{t+r}^p\norm{t-r}^q},
\end{equation}
what guarantees that $\varphi$ is also spherically symmetric.
Introducing new coordinates $u=t+r$ and $v=t-r$ with $u\geq 0$, $|v|\leq u$, and new variables $\psi(u,v):=r\, \varphi(t,r)$ and $H(u,v):=r\, G(t,r)$ we get
\begin{equation}
   4\, \d_v \d_u \psi(u,v) = H(u,v).
\end{equation}
This equation can be first integrated to
\begin{equation} 
   \d_u \psi(u,v) = \frac{1}{4}\int_{-u}^v H(u,v') dv' + (\d_u \psi)(u,-u),
\end{equation}
where the last term vanishes on the basis of the initial conditions
\begin{equation}
  (\d_u \psi)(u,-u) = 
  \frac{1}{2}\varphi(0,u) + \frac{u}{2}(\d_t\varphi)(0,u) + \frac{u}{2}(\d_r\varphi)(0,u) = 0.
\end{equation}
Then we can integrate again and find
\begin{equation} 
   \psi(u,v) = \frac{1}{4} \int_{|v|}^u \left(\int_{-u'}^v H(u',v') dv'\right) du' + \psi(|v|,v),
\end{equation}
where the last term vanishes again because $\psi(|v|,v)=\frac{|v|-v}{2}\varphi\left(\frac{|v|+v}{2},\frac{|v|-v}{2}\right)$ what is equal to either $0\cdot \varphi(v,0)=0$ for $v\geq 0$ or $|v| \varphi(0,|v|)$ for $v<0$, which is zero on the basis of the initial conditions.

The double integral can be relatively easy estimated, because $H$ has a very simple form in variables $(u,v)$
\begin{equation} 
  H(u,v)=\frac{u-v}{2}\;G\left(\frac{u+v}{2},\frac{u-v}{2}\right) =
  \frac{A}{2} \frac{u-v}{\norm{u}^p\norm{v}^q}
\end{equation}

\section{Proofs}
\subsection{Proof of lemma \ref{Lem:decay}}
\begin{proof}
As explained in section \ref{Sec:Idea} we first prove an estimate for the spherically symmetric $\varphi(t,r)$ solving \refa{wave-eq-r} with the source $G(t,r)$ satisfying \refa{source-bound-r}.
With the notation as in section \ref{Sec:Idea} we have
\begin{equation}
  \d_u \psi(u,v) = 
  \frac{A}{8} \int_{-u}^v \frac{u-v'}{\norm{u}^p\norm{v'}^q} dv' =
  \frac{A}{8} \bigg[\frac{u}{\norm{u}^p} \underbrace{\int_{-u}^v \frac{dv'}{\norm{v'}^q}}_{I_1} +
  \frac{1}{\norm{u}^p} \underbrace{\int_{-u}^v \frac{-v' dv'}{\norm{v'}^q}}_{I_2} \bigg]
\end{equation}
We estimate the first integral
\begin{equation}
  I_1= \int_{-u}^v \frac{dv'}{(1+|v'|)^q} \leq 2 \int_{0}^u \frac{dv'}{(1+v')^q} \leq 
  2 \int_{0}^\infty \frac{dv'}{(1+v')^q} = \frac{2}{q-1}
\end{equation}
for $q>1$.
A better estimate can be obtained for the case $v<0$, but we want to have a universal formula for all values of $v$. 
Since $I_1={\cal{O}}(1)$, it is enough to show the same for $I_2$. Better estimates are possible, but will not bring any advantage here. We can write
\begin{equation}
  I_2=\int_{-v}^{u} \frac{v' dv'}{(1+|v'|)^q} = 
  \underbrace{\int_{-v}^{|v|} \frac{v' dv'}{(1+|v'|)^q}}_{=0} + \int_{|v|}^{u} \frac{v' dv'}{(1+|v'|)^q} 
  \leq \int_{0}^{\infty} \frac{v' dv'}{(1+v')^q} = \frac{1}{(p-1)(p-2)}
\end{equation}
for $p>2$. 
Then, we get
\begin{equation}
  \d_u \psi(u,v) \leq 
  \frac{A}{8} \bigg[\frac{u}{\norm{u}^p} \frac{2}{q-1} +
  \frac{1}{\norm{u}^p} \frac{1}{(p-1)(p-2)} \bigg] \leq \frac{B}{\norm{u}^{p-1}}
\end{equation}
with $B:=\frac{1}{8}\left(\frac{2}{q-1}+\frac{1}{(p-1)(p-2)}\right)$.
As next we have
\begin{equation}
  \psi(u,v) = \int_{|v|}^u \d_u \psi(u',v) du' \leq B \int_{|v|}^u \frac{du'}{\norm{u'}^{p-1}}
\end{equation}
and we find
\begin{equation}
  \int_{|v|}^u \frac{du'}{\norm{u'}^{p-1}} = \frac{1}{p-2}
  \left[\frac{1}{\norm{v}^{p-2}} -  \frac{1}{\norm{u}^{p-2}}\right] =
  \frac{1}{p-2} \frac{1}{\norm{v}^{p-2}} \left[1 - \left(\frac{1+|v|}{1+u}\right)^{p-2}\right].
\end{equation}
We use an easy to prove inequality for $\nu>0$ and $0\leq x\leq 1$
\begin{equation}
  1-x^\nu \leq \max(1,\nu) (1-x)
\end{equation}
and get
\begin{equation}
  1 - \left(\frac{1+|v|}{1+u}\right)^{p-2} \leq 
  \max(1,p-2) \left(1 - \frac{1+|v|}{1+u}\right) =
  \max(1,p-2) \frac{2\min(t,r)}{\norm{u}}.
\end{equation}
It leads to
\begin{equation}
  \psi(u,v) \leq C \frac{\min(t,r)}{\norm{v}^{p-2}\norm{u}}
\end{equation}
with $C:=2 B \max\left(1,\frac{1}{p-2}\right)$.
Finally, we find
\begin{equation}
  \varphi(t,r) = \frac{1}{r} \psi(u,v) \leq C\frac{\min(t,r)}{r \norm{v}^{p-2}\norm{u}} 
  \leq \frac{C}{\norm{t-r}^{p-2}\norm{t+r}}.
\end{equation}
Thus we have an estimate for $\varphi(t,r)$ satisfying $\Box \varphi = G$ with $G$ given by \refa{source-bound-r}. Now, we use the comparison theorem and its corollary \ref{Cor:comparison}. Since, by \refa{source-bound}, $|F(t,x)|\leq G(t,|x|)$, we get $|\phi(t,x)|\leq \varphi(t,|x|)$ what, together with the last estimate, finishes the proof of lemma \ref{Lem:decay}.
\end{proof}

\subsection{Proof of lemma \ref{Lem:decay-x}}
\begin{proof}
Again, we first prove an estimate for the spherically symmetric $\varphi(t,r)$ solving \refa{wave-eq-r} with the source $G(t,r)$ which is now
\begin{equation} \label{source-bound-r-x}
  G(t,r)=\frac{A}{\norm{r}^\la\norm{t+r}^p\norm{t-r}^q}.
\end{equation}
In this case we have
\begin{equation}
\begin{split}
  \d_u \psi(u,v) &= 
  \frac{A}{8} \int_{-u}^v \frac{u-v'}{\norm{u-v'}^\la\norm{u}^p\norm{v'}^q} dv' \leq
  \frac{A}{8\norm{u}^p} \int_{-u}^v \frac{dv'}{\norm{u-v'}^{\la-1}\norm{v'}^q} \\ &\leq
  \frac{A}{8\norm{u}^p} \bigg[\underbrace{\int_{0}^u \frac{dv'}{\norm{u+v'}^{\la-1}\norm{v'}^q}}_{I_1}
  +  \underbrace{\int_{0}^{|v|} \frac{dv'}{\norm{u-v'}^{\la-1}\norm{v'}^q}}_{I_2} \bigg]
\end{split}
\end{equation}
We estimate the first integral
\begin{equation}
\begin{split}
  I_1 &= \frac{1}{(1+u)^{\la-1}} \int_{0}^u \frac{dv'}{\left(1+\frac{v'}{1+u}\right)^{\la-1}(1+v')^q}
  \leq \frac{1}{(1+u)^{\la-1}} \int_{0}^u \frac{dv'}{(1+v')^q} \\
  &\leq \frac{1}{(1+u)^{\la-1}} \int_{0}^\infty \frac{dv'}{(1+v')^q}
  = \frac{1}{(q-1)} \frac{1}{(1+u)^{\la-1}}
\end{split}
\end{equation}
for $q>1$. As before, a better estimate can be obtained for the case $v<0$, but we want to have a universal formula for all values of $v$. We estimate the second integral introducing $\mu:=\min(q,\la-1)$
\begin{equation}
\begin{split}
  I_2 &\leq \int_{0}^{|v|} \frac{dv'}{(1+u-v')^\mu(1+v')^\mu}
  = \int_{0}^{|v|} \frac{dv'}{[1+u+v'(u-v')]^\mu} \\
  &= \frac{1}{(1+u)^\mu} \int_{0}^{|v|} \frac{dv'}{\left(1+\frac{v'(u-v')}{1+u}\right)^\mu}
  \leq \frac{1}{(1+u)^\mu}\; 2 \int_{0}^{u/2} \frac{dv'}{\left(1+\frac{v'(u-v')}{1+u}\right)^\mu}
\end{split}
\end{equation}
because $v'(u-v')$ is positive for $0\leq v' \leq u$ and symmetric around $v'=u/2$. Further, using the fact that $v'(u-v')\geq uv'/2$ for $v'\leq u/2$, we have
\begin{equation}
\begin{split}
  I_2 &\leq \frac{2}{(1+u)^\mu} \int_{0}^{u/2} \frac{dv'}{\left(1+\frac{u v'}{2(1+u)}\right)^\mu}
\end{split}
\end{equation}
and the last integral can be split into two and estimated by
\begin{equation}
\begin{split}
  &2 \int_{0}^{u/2} \frac{dv'}{\left(1+\frac{u v'}{2(1+u)}\right)^\mu} =
  4\left(1+\frac{1}{u}\right) \int_{0}^{\frac{u^2}{4(1+u)}} \frac{dw}{\left(1+w\right)^\mu} \\
  &\leq  4 \int_{0}^{\infty} \frac{dw}{\left(1+w\right)^\mu} + 
  \int_{0}^{\frac{u}{1+u}} \frac{dw}{\left(1+\frac{u}{4}w\right)^\mu} 
  \leq \frac{4}{\mu-1} + \frac{u}{1+u} \leq \frac{4}{\mu-1} + 1
\end{split}
\end{equation}
for $\mu>1$. 
Then, we get
\begin{equation}
  \d_u \psi(u,v) \leq 
  \frac{A}{8\norm{u}^p} \bigg[\frac{1}{(q-1)} \frac{1}{(1+u)^{\la-1}} +
  \frac{\frac{4}{\mu-1} + 1}{(1+u)^\mu} \bigg] \leq \frac{B}{\norm{u}^{p+\mu}}
\end{equation}
with $B:=\frac{1}{8}\left(1+\frac{1}{q-1}+\frac{4}{\mu-1}\right) \leq \frac{1}{8}\left(1+\frac{5}{\mu-1}\right)$. As next we have
\begin{equation}
  \psi(u,v) = \int_{|v|}^u \d_u \psi(u',v) du' \leq B \int_{|v|}^u \frac{du'}{\norm{u'}^{p+\mu}}
  = \frac{B}{(p+\mu-1)} \left[\frac{1}{\norm{v}^{p+\mu-1}} -  \frac{1}{\norm{u}^{p+\mu-1}}\right] 
\end{equation}
for $p+\mu>1$, what can be estimated, like in the proof of lemma \ref{Lem:decay}, by
\begin{equation}
  \psi(u,v) \leq \frac{\max(1,p+\mu-1) B}{(p+\mu-1)} \frac{2\min(t,r)}{\norm{v}^{p+\mu-1}\norm{u}}
  \leq 2 B \frac{\min(t,r)}{\norm{v}^{p+\mu-1}\norm{u}}.
\end{equation}
Finally, we find
\begin{equation}
  \varphi(t,r) = \frac{1}{r} \psi(u,v) \leq 2B\frac{\min(t,r)}{r \norm{v}^{p+\mu-1}\norm{u}} 
  \leq \frac{C}{\norm{t-r}^{p+\mu-1}\norm{t+r}},
\end{equation}
where $C:=\frac{1}{4}\left(1+\frac{1}{q-1}+\frac{4}{\mu-1}\right)$.
Thus we have an estimate for $\varphi(t,r)$ satisfying $\Box \varphi = G$ with $G$ given by \refa{source-bound-r-x}. Corollary \ref{Cor:comparison} implies that from \refa{source-bound-x}, i.e. $|F(t,x)|\leq G(t,|x|)$, follows $|\phi(t,x)|\leq \varphi(t,|x|)$ what, together with the last estimate, finishes the proof of lemma \ref{Lem:decay-x}.
\end{proof}

\section{Applications}
\subsection{Lemma \ref{Lem:decay} -- nonlinear wave equation}

We only briefly sketch the main application of lemma \ref{Lem:decay}, referring the reader interested in the details to the literature addressing the concrete problems \cite{John-blowup, Asakura, Strauss-T, Georg-H-K}. One of them is the main step in the proof of existence and decay of solutions to a nonlinear wave equation
\begin{equation} \label{wave-eq-nonlin}
  \Box \phi = F(\phi)
\end{equation}
where the nonlinear term is such that $|F(\phi)|\leq A |\phi|^p$ and $p>1+\sqrt{2}$. The proof usually uses iteration $\Box \phi_{n+1} = F(\phi_n)$ and induction in $n$, where one assumes
\begin{equation}
  |\phi_n| \leq \frac{C_n}{\norm{t+|x|}\norm{t-|x|}^{\la}}
\end{equation}
and must show the same estimate for $\phi_{n+1}$. Then one has
\begin{equation}
  |F(\phi_n)|\leq \frac{A (C_n)^p}{\norm{t+|x|}^p \norm{t-|x|}^{p\,\la}}
\end{equation}
and by lemma \ref{Lem:decay} one gets for $p>2$ and $q:=p\,\la>1$
\begin{equation}
  |\phi_{n+1}| \leq \frac{C_{n+1}}{\norm{t+|x|}\norm{t-|x|}^{p-2}}.
\end{equation}
The iteration closes when $p-2\geq\la$. Choosing the optimal (biggest) value $\la:=p-2$ the remaining conditions reduce to $p\,(p-2)>1$, $p>2$ what gives $p>1+\sqrt{2}$. 
Then, one needs only to show that the sequence $\phi_n$ is Cauchy w.r.t. the norm $\|\phi\|_\Linftx{p} := \|\norm{t+|x|}\norm{t-|x|}^{p-1}\phi(t,x)\|_\LinfTX$ and since the normed space (a weighted-$\Linf$ space) is Banach, the sequence $\phi_n$ converges to a solution of the nonlinear wave equation \refa{wave-eq-nonlin}.

\subsection{Lemma \ref{Lem:decay-x} -- linear wave equation with potential}

Here, we sketch an important application of lemma \ref{Lem:decay-x}, which is the main estimate in the proof of existence and decay of solutions to a wave equation with a potential term (referring the reader to \cite{Strauss-T, Georg-H-K, NS-WaveDecay} for the details) 
\begin{equation} \label{wave-eq-pot}
  \Box \phi + V u = 0
\end{equation}
where the potential is bounded by $|V(x)|\leq V_0/\norm{x}^\la$ with $\la>2$. The proof again uses iteration $\Box \phi_{n+1} = -V \phi_n$ and induction in $n$, where one analogously assumes
\begin{equation}
  |\phi_n| \leq \frac{C_n}{\norm{t+|x|}\norm{t-|x|}^{q-1}}
\end{equation}
and must show the same for $\phi_{n+1}$. Then 
\begin{equation}
  |V \phi_n|\leq \frac{V_0 C_n}{\norm{x}^\la \norm{t+|x|} \norm{t-|x|}^{q-1}}
\end{equation}
and by lemma \ref{Lem:decay-x} one gets for $\la>2$, $q-1>1$ and $p:=1$
\begin{equation}
  |\phi_{n+1}| \leq \frac{C_{n+1}}{\norm{t+|x|}\norm{t-|x|}^\nu},
\end{equation}
where $\nu:=\min(q-1,\la-1)$. The iteration closes when $\nu\geq q-1$, i.e. when $q\leq \la$ (the optimal value is $q:=\la$).
Then, again, it remains only to show that the sequence $\phi_n$ is Cauchy w.r.t. the norm $\|\phi\|_\Linftx{q} := \|\norm{t+|x|}\norm{t-|x|}^{q-1}\phi(t,x)\|_\LinfTX$ and then $\phi_n$ converges to a solution of the wave equation with potential \refa{wave-eq-pot}.

\bibliography{QNMs}
\bibliographystyle{unsrt}

\end{document}